\documentclass[12pt]{article}
\usepackage[height=220mm,width=150mm]{geometry}
\usepackage{array}
\usepackage{color}
\usepackage{mathrsfs}
\usepackage{hyperref}
\usepackage{amssymb}
\usepackage{amsmath}
\usepackage{amsthm}
\usepackage{color}
\allowdisplaybreaks
\usepackage{graphicx}
\usepackage{subfigure}
\usepackage{tikz}
\usetikzlibrary{shapes.arrows,chains,positioning}
\newtheorem{Theorem}{Theorem}[section]
\newtheorem{Lemma}[Theorem]{Lemma}
\newtheorem{Definition}[Theorem]{Definition}

\newtheorem{Remark}[Theorem]{Remark}

\def\nocolor#1{}
\begin{document}
	
	\title{Stable recovery of complex dictionary-sparse signals from phaseless measurements}
	
	\author{Lianxing Xia and Haiye Huo\thanks{Corresponding author.}\\
		\normalsize{Department of Mathematics, School of Mathematics and Computer Sciences},\\ \normalsize{Nanchang University, Nanchang~330031, Jiangxi, China} \\
		\normalsize{Emails: xialianxing@foxmail.com; hyhuo@ncu.edu.cn}
	}
	
	\date{}
	\maketitle

	\textbf{Abstract:}
	Dictionary-sparse phase retrieval, which is also known as phase retrieval with redundant dictionary, aims to reconstruct an original dictionary-sparse signal from its measurements without phase information. It is proved that if the measurement matrix $A$ satisfies null space property (NSP)/strong dictionary restricted isometry property (S-DRIP), then the dictionary-sparse signal can be exactly/stably recovered from its magnitude-only measurements up to a global phase. However, the S-DRIP holds only for real signals. Hence, in this paper, we mainly study the stability of the $\ell_1$-analysis minimization and its generalized $\ell_q\;(0<q\leq1)$-analysis minimization for the recovery of complex dictionary-sparse signals from phaseless measurements. First, we introduce a new $l_1$-dictionary restricted isometry property ($\ell_1$-DRIP) for rank-one and dictionary-sparse matrices, and show that complex dictionary-sparse signals can be stably recovered by magnitude-only measurements via $\ell_1$-analysis minimization provided that the quadratic measurement map $\mathcal{A}$ satisfies $\ell_1$-DRIP. Then, we generalized the $\ell_1$-DRIP condition under the framework of $\ell_q\;(0<q\leq1)$-analysis minimization.
	
	
	
	\section{Introduction}\label{sec:I}
	
	\subsection{Sparse Phase Retrieval}
	The sparse phase retrieval problem \cite{CJL22,HLV24,JOH17,LCS23,LFL25,SEC15,WZG18} is to recover a sparse signal $\textbf{x}_0\in \mathbb{F}^{n}$ from the magnitude of its noisy measurements $y_i=\vert\langle \textbf{a}_i,\textbf{x}_0 \rangle \vert^2+e_i$, $i=1, 2,\dots, m$, where $\textbf{a}_i\in\mathbb{F}^{n}$ stands for the measurement vector, and $e_i\in \mathbb{R}$ represents the noise, $\mathbb{F}\in\left\lbrace \mathbb{R}, \mathbb{C}\right\rbrace$. Note that there exist some trivial ambiguities in sparse phase retrieval, such as time shift, conjugate flip, global phase change, etc. Therefore, the uniqueness in sparse phase retrieval is possible only up to a global phase.
	Let $\mathcal{A}:\mathbb{F}^{n\times n} \rightarrow \mathbb{R}^m$ be a linear map satisfies
	\begin{equation}\label{def:A} \mathcal{A}(X):=\mathcal{A}(\mathbf{x}\mathbf{x}^*)=(\textbf{a}_1^*\mathbf{x}\mathbf{x}^*\textbf{a}_1,\ldots,\textbf{a}_m^*\mathbf{x}\mathbf{x}^*\textbf{a}_m)
		=(\vert\langle \textbf{a}_1,\textbf{x}\rangle|^2,\ldots,\vert\langle \textbf{a}_m,\textbf{x}\rangle|^2).
	\end{equation}
	Set $\textbf{y}=(y_1,\ldots,y_m)^T\in\mathbb{R}^m$, and $\textbf{e}=(e_1,\ldots,e_m)^T\in\mathbb{R}^m$. Then, the goal of sparse phase retrieval is to reconstruct
	$\hat{\textbf{x}}_0:=\{c\textbf{x}_0:\vert c\vert=1,c\in \mathbb{F}\}$ from its noisy measurements
	$\textbf{y}=\mathcal{A}(\textbf{x}_0)+\textbf{e}$.
	
	Wang and Xu \cite{WX2014} showed that the null space property (NSP) is a necessary and sufficient condition to ensure sparse signals can be accurately reconstructed from the
	magnitude-only measurements. Voroninski and Xu \cite{VZ2016} utilized $\ell_{1}$ minimization to solving sparse phase retrieval for real sparse signals without noise as follows:
	\begin{equation}\label{sx}
		\mathop{\min} \limits_{\textbf{x}\in\mathbb{R}^n}\Vert \textbf{x}\Vert_1 \qquad  {\mbox{subject to}}  \qquad \vert A\textbf{x}\vert=\vert A\textbf{x}_0\vert,
	\end{equation}
	where $A=(\textbf{a}_1,\ldots,\textbf{a}_m)^T\in \mathbb{R}^{m\times n}$ is the measurement matrix. It is shown in \cite{VZ2016} that the real $k$-sparse signals can be exactly recovered from the magnitude-only measurements via $\ell_1$ minimization (\ref{sx}), if the measurement matrix $A$ satisfies strong restricted isometry property (SRIP). Moreover, for the noisy setting, Gao, Wang, and Xu \cite{GW2016} proved that SRIP can guarantee a stable recovery of real $k$-sparse signals by solving $\ell_1$ minimization. Xia and Xu \cite{XX2021} extended these results to complex sparse signals by solving $\ell_{1}$ minimization in the following
	\begin{equation}\label{csx}
		\mathop{\min} \limits_{\textbf{x}\in\mathbb{C}^n}\Vert \textbf{x}\Vert_1 \qquad  {\mbox{subject to}}  \qquad \Vert \mathcal{A}(\textbf{x}\textbf{x}^*)-\textbf{y}\Vert_2\le\varepsilon,
	\end{equation}
	where $\textbf{y}=\mathcal{A}(\textbf{x}_0)+\textbf{e}$ with $\|\textbf{e}\|_2\le\varepsilon$. Note that the constrained condition in (\ref{csx}) is non-convex, Xia and Xu \cite{XX2021} replaced (\ref{csx}) with recovering rank-one and sparse matrices by using the lifting technique, that is to say,
	\begin{equation}\label{AX}
		\mathop{\min} \limits_{X\in \mathbb{H}^{n\times n}}\Vert X\Vert_1  \qquad  {\mbox{subject to}} \qquad \Vert\mathcal{A}(X)-\textbf{y}\Vert_2\leq\varepsilon, \quad{\mbox{rank}}(X)=1,
	\end{equation}
	where $X=\textbf{x}\textbf{x}^*$, $\mathbb{H}^{n\times n}$ is the set of Hermitian matrices. It is shown that a complex $k$-sparse signal can be stably reconstructed from $\mathcal{A}$, providing $\mathcal{A}$ satisfies the novel RIP of order $(2,2sk)$. A map $\mathcal{A}$: $\mathbb{H}^{n\times n}\longrightarrow \mathbb{R}^m$ is called to satisfy the RIP of order $(s,k)$, if there exist two constants $\alpha$ and $\beta$, such that
	\begin{equation}\label{CRIP}
		\alpha\Vert X\Vert_{F}\leq\frac{1}{m}\Vert\mathcal{A}(X)\Vert_{1}\leq \beta\Vert X\Vert_{F}
	\end{equation}
	holds for every $X\in\mathbb{H}^{n\times n}$, with rank$(X)\leq s$ and $\Vert X\Vert_{0,2}\leq k$. Here, $\Vert X\Vert_{0,2}$ represents the number of non-zero rows in $X$, and $\Vert X\Vert_{F}=\sqrt{\sum_{j=1}^n\sum_{k=1}^n|X_{j,k}|^2}$. Then, Xia and Xu \cite{XX21TSP} proposed a Sparse PhaseLiftOff model to recover complex sparse signals from quadratic measurements.
	Huang et al. \cite{HSX2023} studied recovery conditions on recovering real (or complex) sparse signals from the magnitude-only measurements with prior knowledge via $\ell_1$ minimization. Xia and Xu \cite{XX2024} investigated the performance of amplitude-based model for complex sparse phase retrieval with noise. Huang and Xu \cite{HX2024} showed the performance of intensity-based estimators for complex phase retrieval and sparse phase retrieval, respectively.
	
	However, all the results mentioned above are based on $\ell_1$-minimization. Xia and Zhou \cite{XZ23} extended the uniqueness results for complex sparse signals \cite{XX2021} to $\ell_q$-minimization:
	\begin{equation*}
		\mathop{\min} \limits_{\mathbf{x}\in\mathbb{C}^{n\times n}}\Vert \textbf{x}\Vert_q^q \qquad  {\mbox{s.t.}}  \qquad  |A\mathbf{x}|^2=|A\mathbf{x}_0|^2,
	\end{equation*}
	where $0<q\le 1$. By utilizing the operator $\mathcal{A}$ and lifting technique, the $\ell_q$-minimization can be reformulated as:
	\begin{equation}\label{add-lp}
		\mathop{\min} \limits_{X\in\mathbb{H}^{n\times n}}\Vert X\Vert_q^q  \qquad  {\mbox{s.t.}}  \qquad \mathcal{A}(X)=\mathcal{A}(X_0),
	\end{equation}
	where $\|X\|_q^q:=\sum_{j=1}^n\sum_{k=1}^n|X_{j,k}|^q$.
	They proved that if a map satisfies $\ell_q$-RIP, then the low-rank and sparse matrices can be uniquely recovered by solving model (\ref{add-lp}). We recall that a map $\mathcal{A}$ has the $\ell_q$-RIP of order $(s,k)$, if there exist constants $\theta_{-},\;\theta_{+}>0$ that satisfy
	\begin{equation*}
		\theta_{-}\Vert X\Vert_{F}^q\leq\Vert\mathcal{A}(X)\Vert_{q}^q\leq \theta_{+}\Vert X\Vert_{F}^q,
	\end{equation*}
	for any $X\in\mathbb{H}^{n\times n}$ with rank$(X)\leq s$ and $\Vert X\Vert_{0,2}\leq k$. Li et al. \cite{LWL25} studied the theoretical guarantee on the $\ell_q$-minimization for nearly-sparse complex signals that are corrupted with noise.
	
	\subsection{Dictionary-Sparse Phase Retrieval}
	
	The conclusions mentioned above hold for signals which are sparse in the standard coordinate basis. However, in practice, many signals have special structures, such as block-sparse \cite{ZSZ22}, weighted-sparse \cite{Huo22,HX24,ZY2020}, or dictionary-sparse \cite{CE2011,G2016}, etc. In this paper, we focus on studying the phase retrieval problem of dictionary-sparse signals, which are not sparse in orthonormal bases, but rather in overcomplete dictionaries.
	Let $D\in \mathbb{F}^{n\times N}$ be a redundant dictionary, then the signal $\mathbf{x}_0\in\mathbb{F}^{n}$ is called dictionary $k$-sparse, if there exists a $k$-sparse signal $\mathbf{z}_0\in\mathbb{F}^{N}$, such that $\mathbf{x}_0=D\mathbf{z}_0$. Gao \cite{G2016} first established a theoretical framework for recovering a dictionary $k$-sparse signal $\textbf{x}_0= D\textbf{z}_0$ from the magnitude-only measurements $\textbf{y}=\vert AD\textbf{z}_0\vert$, where $\textbf{z}_0$ is $k$-sparse. Mathematically, it can be rewritten as $\ell_1$-analysis model:
	\begin{equation}\label{RD}
		\mathop{\min} \limits_{\textbf{x}\in\mathbb{F}^n}\Vert D^*\textbf{x}\Vert_1 \qquad {\mbox{subject to}} \qquad \Vert\vert A\textbf{x}\vert-\vert A\textbf{x}_0\vert\Vert_2\leq\varepsilon,
	\end{equation}
	where $D^*$ is the adjoint conjugate of $D$.
	When $\mathbb{F}=\mathbb{R}$, Gao $\cite{G2016}$ proposed a new concept called strong dictionary restricted isometry property (S-DRIP), and proved that real dictionary-sparse signals can be stably recovered by the phaseless measurements, if the measurement matrix $A$ satisfies S-DRIP. Cao and Huang \cite{CH2022} generalized the results proposed in \cite{G2016} to $\ell_q \;(0<q\le 1)$ minimization case. However, to the best of our knowledge, there are no relevant results on the stable recovery of dictionary-sparse signals in the complex number filed. To fill this gap, in this paper, we generalize the results proposed in $\cite{XX2021,XZ23}$ to the complex dictionary-sparse signals.
	
	The rest of the paper is organized as follows. In Sect.~\ref{sec2}, we investigate a sufficient condition, the $l_1$-DRIP, on quadratic measurements for exact recovery of complex dictionary-sparse signals by solving $l_1$-analysis minimization. Then, we extend the results to the $\ell_q\;(0<q\le 1)$-analysis model in Sect.~\ref{sec3}. Finally, we conclude the paper.
	
	\section{$\ell_1$-DRIP}\label{sec2}
	In this section, we consider to recover a complex dictionary-sparse signal $\textbf{x}_0$ from its phaseless measurements $\mathbf{y}=\mathcal{A}(\textbf{x}_0)+\textbf{e}=\mathcal{A}(D\textbf{z}_0)+\textbf{e}$, where $\|\textbf{e}\|_2\le\varepsilon$. The $\ell_1$-analysis model can be represented as
	\begin{equation}\label{Dx}
		\mathop{\min} \limits_{\textbf{x}\in\mathbb{C}^n}\Vert D^*\textbf{x}\Vert_1  \qquad {\mbox{subject to}} \qquad \Vert\mathcal{A}(\textbf{x})-\textbf{y}\Vert_2\leq\varepsilon.
	\end{equation}
	Motivated by $\cite{XX2021}$, we turn to recover rank-one and sparse matrices by lifting $(\ref{Dx})$, that is to say,
	\begin{equation}\label{DX}
		\mathop{\min} \limits_{X\in\mathbb{H}^{n\times n}}\Vert D^*XD\Vert_1 \quad\; {\mbox{subject to}} \quad \; \Vert\mathcal{A}(X)-\textbf{y}\Vert_2\leq\varepsilon,\quad {\mbox{rank}}(X)=1.
	\end{equation}
	
	Then, we extend the concept of RIP (\ref{CRIP}) for low-rank and sparse matrices to dictionary-sparse case.
	\begin{Definition}[$\ell_1$-DRIP]\label{DRIP}
		For a dictionary $D\in\mathbb{F}^{n\times N}$ and $\mathcal{A}$: $\mathbb{H}^{n\times n}\longrightarrow \mathbb{R}^m$. $\mathcal{A}$ is said to satisfy the $\ell_1$-DRIP of order $(s,k)$ with positive parameters $\alpha$ and $\beta$, if
		\begin{equation}\label{CDRIP}
			\alpha\lVert DZD^*\rVert_{F}\leq \frac{1}{m}\lVert\mathcal{A}(DZD^*)\rVert_{1}\leq \beta\lVert DZD^*\rVert_{F}
		\end{equation}
		holds for all $\mathit{Z}\in\mathbb{H}^{N\times N}$ with rank($Z$)$\leq s$ and $\lVert \mathit{Z}\rVert_{0,2}\leq k$.
	\end{Definition}
	
	Note that our proposed $\ell_1$-DRIP is a natural generalization of the RIP introduced in (\ref{CRIP}). In fact, for fixed $DZD^*\in\mathbb{H}^{n\times n}$, if $m\gtrsim k\log(n/k)$, then any map $\mathcal{A}: \mathbb{H}^{n\times n}\rightarrow\mathbb{R}^{m}$ obeying
	\begin{equation}
		\mathbb{P}\left[0.12\Vert DZD^*\Vert_{F}\leq\frac{1}{m}\Vert \mathcal{A}(DZD^*)\Vert_1\leq 2.45 \Vert DZD^*\Vert_{F}\right] \geq 1-2e^{-\xi m}
	\end{equation}
	($\xi$ is a positive constant) will satisfy the $\ell_1$-DRIP with high probability. This can be derived by a standard covering argument (see the proof of Theorem 1.2 in \cite{XX2021} for more details).
	
	Next, we show that complex dictionary-sparse signals can be exactly reconstructed from (\ref{DX}) up to a global phase, if $\mathcal{A}$ satisfies the $\ell_1$-DRIP of order $(2,2rk)$, where $r>0$ is chosen appropriately.
	\begin{Theorem}\label{recovery}
		Suppose that $D\in \mathbb{C}^{n\times N}$ is a dictionary, and $\mathbf{x}_{0}\in\mathbb{C}^n$ is a dictionary $k$-sparse signal. If $\mathcal{A(\cdot)}$ satisfies the $\ell_1$-DRIP of order $(2,2rk)$ with
		\begin{equation}\label{Ca}
			\alpha-\frac{4\beta}{\sqrt{r}}-\frac{\beta}{r}>0.
		\end{equation}
		Then, the solution $\hat{\mathbf{x}}$ to model (\ref{Dx}) satisfies
		\begin{equation}
			\lVert \hat{\mathbf{x}}{(\hat{\mathbf{x}})}^*-\mathbf{x}_{0} {\mathbf{x}_{0}}^* \rVert_{F}\leqslant \frac{2C\varepsilon}{\sqrt{m}},
		\end{equation}
		where
		\begin{eqnarray*}
			C=\frac{\frac{1}{r}+\frac{4}{\sqrt{r}}+1}{\alpha-\frac{4\beta}{\sqrt{r}}-\frac{\beta}{r}}.
		\end{eqnarray*}
		In addition, we have
		\begin{eqnarray*}
			\mathop{\min} \limits_{c\in\mathbb{C},\vert c\vert=1}\Vert c\cdot{\hat{\mathbf{x}}}-\mathbf{x}_{0}\Vert_2\leqslant\frac{2\sqrt{2}C\varepsilon}{\sqrt{m}\Vert \mathbf{x}_{0}\Vert_2}.
		\end{eqnarray*}	
	\end{Theorem}
	In order to prove Theorem~\ref{recovery}, we first give two useful lemmas.
	
	\begin{Lemma}\cite{CZ2014,G2016}\label{vu}
		Assume that $\mathbf{v}\in\mathbb{R}^q$ satisfies $\Vert \mathbf{v}\Vert_{1}\leq k\mu$ and $\Vert \mathbf{v}\Vert_{\infty}\leq\mu$, where $\mu>0$ and $k$ is a positive integer. Set
		\[
		U(\mu,k,\mathbf{v}):=\{\mathbf{u}\in\mathbb{R}^q:\;{\rm{supp}}\left(\mathbf{u}\right)\subseteq {\rm{supp}}\left(\mathbf{v}\right),\|\mathbf{u}\|_0\le k,\;\Vert\mathbf{u}\Vert_{1}=\Vert \mathbf{v}\Vert_{1},\; \Vert\mathbf{u}\Vert_{\infty}\leq\mu\}.
		\]
		Then $\mathbf{v}$ can be represented by the convex combination of $k$-sparse vectors $\{\mathbf{u}_i\}_{i=1}^{P}$, that is to say,
		\begin{equation*}
			\mathbf{v}=\sum_{i=1}^{P}\lambda_i\mathbf{u}_i,
		\end{equation*}
		where $ 0\leq\lambda_i\leq1,\;\sum_{i=1}^{P}\lambda_i=1$, and $\mathbf{u}_i\in U(\mu,k,\mathbf{v})$ for $i=1,2,\cdots,P$.
	\end{Lemma}	
	\begin{Lemma}\cite{XX2021}\label{xy}
		For any $\mathbf{u},\mathbf{v}\in\mathbb{C}^p$, suppose that $\langle \mathbf{u},\mathbf{v}\rangle\geq0 $, then we have
		\begin{eqnarray*}
			\Vert \mathbf{u}\mathbf{u}^*-\mathbf{v}\mathbf{v}^*\Vert_{F}^2\geq\frac{1}{2}\Vert \mathbf{u}\Vert_{2}^2\Vert\mathbf{u}-\mathbf{v}\Vert_{2}^2,
		\end{eqnarray*}
		and
		\begin{eqnarray*}
			\Vert \mathbf{u}\mathbf{u}^*-\mathbf{v}\mathbf{v}^*\Vert_{F}^2\geq\frac{1}{2}\Vert \mathbf{v}\Vert_{2}^2\Vert \mathbf{u}-\mathbf{v}\Vert_{2}^2.
		\end{eqnarray*}
	\end{Lemma}
	Now, we are ready to prove Theorem~\ref{recovery}.
	\begin{proof}[Proof of Theorem~\ref{recovery}]
		Suppose that the dictionary $D\in \mathbb{C}^{n\times N}$ is normalized, i.e., $DD^*=I_{n}$, and $\|D^*YD\|_F=\|Y\|_F$. Let $\hat{\textbf{x}}$ be a solution to $(\ref{Dx})$, then
		$e^{it}\hat{\textbf{x}}$ is also a solution to $(\ref{Dx})$ for all $t\in\mathbb{R}$. In order to use Lemma~\ref{xy} in the subsequent proof, without loss of generality, we suppose that
		\begin{equation*}
			\left\langle D^*\hat{\textbf{x}},D^*\textbf{x}_0\right\rangle\in \mathbb{R} \qquad and \qquad \left\langle D^*\hat{\textbf{x}},D^*\textbf{x}_0\right\rangle\geq0.
		\end{equation*}
		By using the lifting technique, the $\ell_1$-analysis model $(\ref{Dx})$ can be transformed into $\left(\ref{DX}\right)$. Note that $\hat{X}$ is the solution to $\left(\ref{DX}\right)$ if and only if
		$\hat{X}=\hat{\textbf{x}}(\hat{\textbf{x}})^\ast$.
		
		Denote $H:=\hat{X}-X_0$, and $X_0:=\textbf{x}_0{\textbf{x}_0}^\ast$. Let $T_0$ be the support of ${D}^*\textbf{x}_0$. Then, we partition $T_0^c$ into subsets $T_1,T_2,\cdots$ with
		$|T_i|=rk,\;i=1,2\cdots$ according to the non-increasing rearrangement of $({D}^*\hat{\textbf{x}})_{T_0^c}$ in magnitude. Let $T_{01}:=T_0\cup T_1$,\; $D^*\bar{H}D:=(D^*\hat{\textbf{x}})_{T_{01}}(D^*\hat{\textbf{x}})^*_{T_{01}}-(D^*\textbf{x}_0)_{T_{01}}(D^*\textbf{x}_0)^*_{T_{01}}$, and $\left( D^*HD\right)_{T_i,T_j}:=\left( D^*\hat{\textbf{x}}\right)_{T_i}\left(D^*\hat{\textbf{x}}\right)^*_{T_j}-\left( D^*{\textbf{x}_0}\right)_{T_i}\left(D^*{\textbf{x}_0}\right)^*_{T_j}$.
		By the construction of $T_i$, we obtain
		\begin{equation}\label{add:d1}
			\Vert(D^*\hat{\textbf{x}})_{T_i}\Vert_{2}\leq\frac{\Vert(D^*\hat{\textbf{x}})_{T_{i-1}}\Vert_{1}}{\sqrt{rk}}, {\mbox{for}}\,\;i\geq 2.
		\end{equation}
		Note that
		\begin{equation}\label{add:11}
			\Vert H\Vert_{\mathit{F}}=\Vert D^{\ast}HD \Vert_{\mathit{F}}\leq \Vert D^*HD-D^*\bar{H}D\Vert_{\mathit{F}}+\Vert D^*\bar{H}D\Vert_{\mathit{F}}.
		\end{equation}	
		Next, we use two steps to estimate the upper bounds for $\Vert D^*HD-D^*\bar{H}D\Vert_{\mathit{F}}$ and $\Vert D^*\bar{H}D\Vert_{\mathit{F}}$, respectively.
		
		\textbf{Step 1}: Estimate the upper bound of $\Vert D^*HD-D^*\bar{H}D\Vert_{\mathit{F}}$.
		
		After simple calculation, it is easy to obtain
		\begin{align}\label{Ti}
			&\Vert D^*HD-D^*\bar{H}D\Vert_{\mathit{F}} \nonumber\\
			\leq&\sum_{i\geq2}\sum_{j\geq2}\Vert\left( D^*HD\right)_{T_i,T_j}\Vert_{\mathit{F}}+\sum_{i=0,1}\sum_{j\geq2}\Vert\left( D^*HD\right)_{T_i,T_j}\Vert_{\mathit{F}}
			+\sum_{j=0,1}\sum_{i\geq2}\Vert \left( D^*HD\right)_{T_i,T_j} \Vert_{\mathit{F}}\nonumber\\
			=&\sum_{i\geq2}\sum_{j\geq2}\Vert \left( D^*HD\right)_{T_i,T_j}\Vert_{\mathit{F}}+2\sum_{i=0,1}\sum_{j\geq2}\Vert \left( D^*HD\right)_{T_i,T_j}\Vert_{\mathit{F}}.
		\end{align}	
		
		First, we focus on estimating $\sum_{i\geq2}\sum_{j\geq2}\Vert \left( D^*HD\right)_{T_i,T_j}\Vert_{\mathit{F}}$ in (\ref{Ti}).
		
		Since $\Vert D^*\hat{X}D\Vert_{1}\leq\Vert D^*{X_0}D\Vert_{1}$, we get
		\begin{align}
			\Vert D^*HD-(D^*HD)_{T_0,T_0}\Vert_{1}&=\Vert D^*\hat{X}D-(D^*\hat{X}D)_{T_0,T_0} \Vert_{1}\nonumber\\
			&\leq\Vert D^*{X_0}D\Vert_{1}-\Vert(D^*\hat{X}D)_{T_0,T_0}\Vert_{1}\nonumber\\
			&\leq\Vert D^*X_0D-(D^*\hat{X}D)_{T_0,T_0} \Vert_{1}\nonumber\\
			&=\Vert(D^*HD)_{T_0,T_0}\Vert_{1}.\label{add:d2}
		\end{align}
		Applying (\ref{add:d1}) and (\ref{add:d2}), we have
		\begin{align}\label{TiF}
			\sum_{i\geq2}\sum_{j\geq2}\Vert\left( D^*HD\right)_{T_i,T_j}\Vert_{\mathit{F}}
			&=\sum_{i\geq2}\sum_{j\geq2}\Vert\left(D^*\hat{\textbf{x}} \right)_{T_i}\Vert_{2}\cdot\Vert\left(D^*\hat{\textbf{x}} \right)_{T_j}\Vert_{2}\nonumber\\
			&=\left(\sum_{i\geq2}\Vert\left(D^*\hat{\textbf{x}} \right)_{T_i}\Vert_{2}\right)^2\leq\frac{1}{rk}\Vert\left(D^*\hat{\textbf{x}} \right)_{T_0^c} \Vert_{1}^2\nonumber\\
			&=\frac{1}{rk}\Vert\left( D^*HD\right)_{T_0^c,T_0^c}\Vert_{1}\leq\frac{1}{rk}\Vert\left( D^*HD\right)_{T_0,T_0}\Vert_{1}\nonumber\\
			&\leq\frac{1}{r}\Vert\left(D^*HD\right)_{T_0,T_0}\Vert_{F}\leq\frac{1}{r}\Vert D^*\bar{H}D\Vert_{\mathit{F}}.
		\end{align}
		
		Next, we estimate $\sum_{i=0,1}\sum_{j\geq2}\Vert\left( D^*HD\right)_{T_i,T_j}\Vert_{F}$ in (\ref{Ti}).
		
		By (\ref{add:d1}), for $i=0,1$, we obtain
		\begin{align}\label{xi}
			\sum_{j\geq2}\Vert\left( D^*HD\right)_{T_i,T_j}\Vert_{\mathit{F}}
			&=\Vert\left(D^*\hat{\textbf{x}}\right)_{T_i}\Vert_{2}\cdot\sum_{j\geq2}\Vert\left(D^*\hat{\textbf{x}} \right)_{T_j}\Vert_{2}\nonumber\\
			&\leq\frac{1}{\sqrt{rk}}\Vert\left(D^*\hat{\textbf{x}} \right)_{T_i}\Vert_{2}\Vert\left(D^*\hat{\textbf{x}} \right)_{T_0^c}\Vert_{1}.
		\end{align}
		Due to $\Vert\left(D^*\hat{\textbf{x}} \right)\Vert_{1}\leq\Vert\left(D^*\textbf{x}_0 \right)\Vert_{1}$, we have
		\begin{align}
			\Vert\left(D^*\hat{\textbf{x}} \right)_{T_0^c}\Vert_{1}&\leq\Vert D^*\textbf{x}_0\Vert_{1}-\Vert\left(D^*\hat{\textbf{x}} \right)_{T_0}\Vert_{1}
			\leq\Vert\left(D^*\hat{\textbf{x}} \right)_{T_{0}}-D^*\textbf{x}_0\Vert_{1}\nonumber\\
			&\leq\sqrt{k}\Vert\left(D^*\hat{\textbf{x}} \right)_{T_{0}}-D^*\textbf{x}_0\Vert_{2}\leq\sqrt{k}\Vert\left(D^*\hat{\textbf{x}} \right)_{T_{01}}-D^*\textbf{x}_0\Vert_{2}.\label{add:d3}
		\end{align}
		Substituting (\ref{add:d3}) into (\ref{xi}), yields
		\begin{align}\label{add:d4}
			\sum_{j\geq2}\Vert\left( D^*HD\right)_{T_i,T_j}\Vert_{\mathit{F}}
			\leq\frac{1}{\sqrt{r}}\Vert\left(D^*\hat{\textbf{x}} \right)_{T_i}\Vert_{2}\Vert\left(D^*\hat{\textbf{x}} \right)_{T_{01}}-D^*\textbf{x}_0\Vert_{2}, \;{\mbox{for}}\;\; i=0, 1.
		\end{align}
		
		Plugging (\ref{TiF}) and (\ref{add:d4}) into $\left( \ref{Ti}\right) $, we get
		\begin{align}
			\Vert D^*HD-D^*\bar{H}D\Vert_{\mathit{F}}
			&\leq\frac{1}{r}\Vert D^*\bar{H}D\Vert_{\mathit{F}}+\frac{2\sqrt{2}}{\sqrt{r}}\Vert\left(D^*\hat{\textbf{x}}\right)_{T_{01}}\Vert_{2}
			\Vert\left(D^*\hat{\textbf{x}} \right)_{T_{01}}-D^*\textbf{x}_0\Vert_{2}\nonumber\\
			&\leq\left(\frac{1}{r} +\frac{4}{\sqrt{r}}\right)\Vert D^*\bar{H}D\Vert_{\mathit{F}},\label{add:d5}
		\end{align}
		where we use Lemma $\ref{xy}$ in the last inequality.
		
		\textbf{Step 2}: Estimate the upper bound of $\Vert D^*\bar{H}D\Vert_{\mathit{F}}$.
		
		Since
		\begin{eqnarray*}
			\Vert\mathcal{A}[D(D^*HD)D^*]\Vert_{2}=\Vert\mathcal{A}(H)\Vert_{2}\leq\Vert\mathcal{A}(\hat{X})-\textbf{y}\Vert_{2}+\Vert\mathcal{A}\left(X_0\right)-\textbf{y} \Vert_{2}\leq2\varepsilon,
		\end{eqnarray*}
		we get
		\begin{align}\label{mA}
			\nonumber
			\frac{2\varepsilon}{\sqrt{m}}&\geq\frac{1}{\sqrt{m}}\Vert \mathcal{A}[D(D^*HD)D^*]\Vert_{2}\geq\frac{1}{m}\Vert \mathcal{A}[D(D^*HD)D^*]\Vert_{1}\\
			&\geq\frac{1}{m}\Vert \mathcal{A}[D(D^*\bar{H}D)D^*]\Vert_{1}-\frac{1}{m}\Vert \mathcal{A}[D(D^*HD -D^*\bar{H}D)D^*]\Vert_{1}.
		\end{align}
		To obtain a lower bound of $\frac{1}{m}\Vert \mathcal{A}[D(D^*\bar{H}D)D^*]\Vert_{1}-\frac{1}{m}\Vert \mathcal{A}[D(D^*HD -D^*\bar{H}D)D^*]\Vert_{1}$, we estimate the bound of $\frac{1}{m}\Vert \mathcal{A}[D(D^*\bar{H}D)D^*]\Vert_{1}$ from below and $\frac{1}{m}\Vert \mathcal{A}[D(D^*HD-D^*\bar{H}D)D^*]\Vert_{1}$ from above, respectively.
		Notice that
		\[
		{\mbox{rank}}(D^*\bar{H}D)\leq2,\; {\mbox{and}}\;\; \Vert D^*\bar{H}D\Vert_{0,2}\leq(r+1)k.
		\]
		By the definition of $\ell_1$-DRIP (\ref{CDRIP}), yields
		\begin{eqnarray}\label{cm}
			\frac{1}{m}\Vert \mathcal{A}[D(D^*\bar{H}D)D^*]\Vert_{1}\geq \alpha\Vert D(D^*\bar{H}D)D^*\Vert_{F}=\alpha\Vert D^*\bar{H}D\Vert_{F}.
		\end{eqnarray}
		Due to
		\begin{align*}
			D^*HD-D^*\bar{H}D=&(D^*HD)_{T_0,T_{01}^c}+(D^*HD)_{T_{01}^c,T_0}+(D^*HD)_{T_1,T_{01}^c}\\
			&\quad+(D^*HD)_{T_{01}^c,T_1}+(D^*HD)_{T_{01}^c,T_{01}^c},
		\end{align*}
		we obtain
		\begin{align}
			\frac{1}{m}\Vert\mathcal{A}[D(D^*HD -D^*\bar{H}D)D^*]\Vert_{1}
			&\leq\frac{1}{m}\left\|\mathcal{A}\left[D((D^*HD)_{T_0,T_{01}^c}+(D^*HD)_{T_{01}^c,T_0})D^*\right]\right\|_{1}\nonumber\\
			&\;+\frac{1}{m}\left\|\mathcal{A}\left[D((D^*HD)_{T_1,T_{01}^c}+(D^*HD)_{T_{01}^c,T_1})D^*\right]\right\|_{1}\nonumber\\
			&\quad+\frac{1}{m}\left\|\mathcal{A}\left[D(D^*HD)_{T_{01}^c,T_{01}^c}D^*\right]\right\|_{1}.\label{add:d8}
		\end{align}
		By the definition of $\ell_1$-DRIP (\ref{CDRIP}) and (\ref{add:d4}), for $i=0,1$, we get
		\begin{align}
			&\frac{1}{m}\left\|\mathcal{A}\left[D((D^*HD)_{T_i,T_{01}^c}+(D^*HD)_{T_{01}^c,T_i})D^*\right]\right\|_{1}\nonumber\\
			\leq&\frac{1}{m}\sum_{j\geq2}\left\|\mathcal{A}\left[D((D^*HD)_{T_i,T_j}+(D^*HD)_{T_j,T_i})D^*\right]\right\|_{1}\nonumber\\
			\leq&\beta\sum_{j\geq2}\Vert D((D^*HD)_{T_i,T_j}+(D^*HD)_{T_j,T_i})D^*\Vert_{\mathit{F}}\nonumber\\
			=&\beta\sum_{j\geq2}\Vert(D^*HD)_{T_i,T_j}+(D^*HD)_{T_j,T_i}\Vert_{\mathit{F}}\nonumber\\
			\leq& \beta\sum_{j\geq2}\left(\Vert\left(D^*\hat{\textbf{x}}\right)_{T_i}(D^*\hat{\textbf{x}})_{T_j}^*\Vert_{\mathit{F}}
			+\Vert \left(D^*\hat{\textbf{x}}\right)_{T_j}(D^*\hat{\textbf{x}})_{T_i}^*\Vert_{\mathit{F}}\right) \nonumber\\
			=&2\beta\sum_{j\geq2}\Vert\left(D^*\hat{\textbf{x}}\right)_{T_i} \Vert_{2}\Vert\left(D^*\hat{\textbf{x}}\right)_{T_j}  \Vert_{2}\nonumber\\
			\leq&\frac{2\beta}{\sqrt{r}}\Vert\left(D^*\hat{\textbf{x}}\right)_{T_i} \Vert_{2}\Vert\left(D^*\hat{\textbf{x}} \right)_{T_{01}}-D^*\textbf{x}_0\Vert_{2}.\label{HRIP}
		\end{align}
		Next, we turn to bound $\frac{1}{m}\Vert\mathcal{A}[D(D^*HD)_{T_{01}^c,T_{01}^c}D^*]\Vert_{1}$. Note that
		\begin{equation*}
			(D^*HD)_{T_{01}^c,T_{01}^c}=\left(D^*\hat{\textbf{x}} \right)_{T_{01}^c}\left(D^*\hat{\textbf{x}}\right)_{T_{01}^c}^*,
		\end{equation*}
		and
		\[
		\Vert\left(D^*\hat{\textbf{x}}\right)_{T_{01}^c}\Vert_{\infty}\leq\frac{\Vert\left(D^*\hat{\textbf{x}}\right)_{T_{1}}\Vert_{1}}{rk}.
		\]
		Denote
		\[
		\mu:=\max\left\{\frac{\Vert\left(D^*\hat{\textbf{x}}\right)_{T_{1}}\Vert_{1}}{rk}, \frac{\Vert\left(D^*\hat{\textbf{x}}\right)_{T_{01}^c}\Vert_{1}}{rk}\right\}.
		\]
		Suppose that $\Psi:={\rm{diag}}(Ph(\left(D^*\hat{\textbf{x}}\right)_{T_{01}^c}))$ is a diagonal matrix whose elements are the phase of $\left(D^*\hat{\textbf{x}} \right)_{T_{01}^c}$, that is to say,
		$\Psi^{-1}\left(D^*\hat{\textbf{x}}\right)_{T_{01}^c}$ is a real vector. By Lemma $\ref{vu}$, we get
		\begin{equation}\label{add:d6}
			\Psi^{-1}\left(D^*\hat{\textbf{x}}\right)_{T_{01}^c}=\sum_{i=1}^{P}\lambda_i\textbf{u}_i,
		\end{equation}
		where $0\leq\lambda_i\leq1,\;\sum_{i=1}^{P}\lambda_i=1,$\; $\textbf{u}_i$ is $rk$-sparse, and
		$\Vert \textbf{u}_i\Vert_{1}=\Vert (D^*\hat{\textbf{x}})_{T_{01}^c}\Vert_{1},$ and $\Vert \textbf{u}_i\Vert_{\infty}\leq\mu$.
		Hence,
		\begin{equation*}
			\Vert \textbf{u}_i\Vert_{2}\leq\sqrt{\Vert \textbf{u}_i\Vert_{1}\Vert\textbf{u}_i\Vert_{\infty}}\leq\sqrt{\mu\Vert \left(D^*\hat{\textbf{x}} \right)_{T_{01}^c}\Vert_{1}}.
		\end{equation*}
		When $\mu=\frac{\Vert\left(D^*\hat{\textbf{x}} \right)_{T_{1}}\Vert_{1}}{rk}$, we obtain
		\begin{align*}
			\Vert\textbf{u}_i\Vert_{2}&\leq\sqrt{\frac{\Vert\left(D^*\hat{\textbf{x}}\right)_{T_{1}}\Vert_{1}\Vert\left(D^*\hat{\textbf{x}} \right)_{T_{01}^c}\Vert_{1}}{rk}}
			=\sqrt{\frac{\Vert\left(D^*HD\right)_{T_1,T_{01}^c}\Vert_{1}}{rk}}\\
			&\leq\sqrt{\frac{\Vert D^*HD-\left( D^*HD\right)_{T_0,T_0}\Vert_{1} }{rk}}\leq\sqrt{\frac{\Vert\left( D^*HD\right)_{T_0,T_0}\Vert_{1}}{rk}}\\
			&\leq\sqrt{\frac{\Vert\left( D^*HD\right)_{T_0,T_0}\Vert_{F}}{r}}\leq\sqrt{\frac{\Vert\left( D^*\bar{H}D\right)\Vert_{F}}{r}}.
		\end{align*}
		When $\mu=\frac{\Vert\left(D^*\hat{\textbf{x}} \right)_{T_{01}^c}\Vert_{1}}{rk}$, we have
		\begin{align*}
			\Vert\textbf{u}_i\Vert_{2}&\leq\sqrt{\frac{\Vert\left(D^*\hat{\textbf{x}} \right)_{T_{01}^c}\Vert_{1}\Vert\left(D^*\hat{\textbf{x}} \right)_{T_{01}^c}\Vert_{1}}{rk}}
			=\sqrt{\frac{\Vert\left(D^*HD\right)_{T_{01}^c,T_{01}^c}\Vert_{1}}{rk}}\\
			&\leq\sqrt{\frac{\Vert D^*HD-\left(D^*HD\right)_{T_0,T_0}\Vert_{1}}{rk}}\leq\sqrt{\frac{\Vert\left(D^*HD\right)_{T_0,T_0}\Vert_{1}}{rk}}\\
			&\leq\sqrt{\frac{\Vert\left(D^*HD\right)_{T_0,T_0}\Vert_{F}}{r}}\leq\sqrt{\frac{\Vert\left(D^*\bar{H}D\right)\Vert_{F}}{r}}.
		\end{align*}
		Therefore, we have
		\begin{equation}\label{ua}
			\Vert\textbf{u}_i\Vert_{2}\leq\sqrt{\frac{\Vert\left( D^*\bar{H}D\right)\Vert_{F}}{r}},\;\;{\mbox{for}}\;\;i=1,2,\cdots,P.
		\end{equation}
		By (\ref{add:d6}), we know
		\begin{align}
			\left(D^*HD\right)_{T_{01}^c,T_{01}^c}&=\left(D^*\hat{\textbf{x}}\right)_{T_{01}^c}\left(D^*\hat{\textbf{x}}\right)_{T_{01}^c}^*
			=\left(\sum_{i=1}^{P}\lambda_i\Psi\textbf{u}_i\right)\left(\sum_{i=1}^{P}\lambda_i\Psi\textbf{u}_i\right)^*\nonumber\\
			&=\sum_{i<j}\lambda_i\lambda_j\Psi\left(\textbf{u}_i\textbf{u}_j^*+\textbf{u}_j\textbf{u}_i^*\right)\Psi^{-1}+\sum_{i}\lambda_i^2\Psi\textbf{u}_i\textbf{u}_{i}^*\Psi^{-1}.\label{add:d7}
		\end{align}
		Combining $\ell_1$-DRIP (\ref{CDRIP}), (\ref{ua}) and (\ref{add:d7}), we have
		\begin{align}
			\frac{1}{m}\left\|\mathcal{A}\left[D(D^*HD)_{{T_{01}^c},T_{01}^c}D^*\right]\right\|_{1}
			&\leq\beta\sum_{i<j}\lambda_i\lambda_j\Vert\textbf{u}_i\textbf{u}_j^*+\textbf{u}_j\textbf{u}_i^*\Vert_{F}
			+\beta\sum_{i}\lambda_i^2\Vert\textbf{u}_i\textbf{u}_i^*\Vert_{F}\nonumber\\
			&\leq\beta\left(2\sum_{i<j}\lambda_i\lambda_j\Vert\textbf{u}_i\Vert_{2}\Vert\textbf{u}_j\Vert_{2}+\sum_{i}\lambda_i^2\Vert\textbf{u}_i\Vert_{2}^2\right)\nonumber\\
			&\leq\frac{\beta\Vert\left(D^*\bar{H}D\right)\Vert_{F}}{r}\left(\sum_{i}\lambda_i\right)^2=\frac{\beta\Vert\left( D^*\bar{H}D\right)\Vert_{F}}{r}.\label{add:d9}
		\end{align}
		Substituting (\ref{HRIP}) and (\ref{add:d9}) into (\ref{add:d8}), we have
		\begin{align}\label{aF}
			&\frac{1}{m}\Vert\mathcal{A}[D(D^*HD-D^*\bar{H}D)D^*]\Vert_{1}\nonumber\\
			\leq&\frac{2\beta}{\sqrt{r}}\Big(\Vert\left(D^*\hat{\textbf{x}}\right)_{T_0}\Vert_{2}+\Vert\left(D^*\hat{\textbf{x}}\right)_{T_1}\Vert_{2}\Big)
			\Vert\left(D^*\hat{\textbf{x}}\right)_{T_{01}}-D^*\textbf{x}_0\Vert_{2}+\frac{\beta\Vert\left(D^*\bar{H}D\right)\Vert_{F}}{r}\nonumber\\
			\leq&\frac{2\sqrt{2}\beta}{\sqrt{r}}\Vert\left( D^*\hat{\textbf{x}}\right)_{T_{01}}
			\Vert_{2}\Vert\left(D^*\hat{\textbf{x}}\right)_{T_{01}}-D^*\textbf{x}_0\Vert_{2}+\frac{\beta\Vert\left(D^*\bar{H}D\right)\Vert_{F}}{r}\nonumber\\
			\leq&\beta\left(\frac{4}{\sqrt{r}}+\frac{1}{r}\right)\Vert\left(D^*\bar{H}D\right)\Vert_{F},
		\end{align}
		where we use Lemma~\ref{xy} in the last step. Combining $\left(\ref{mA}\right)$, $(\ref{cm})$, and $\left(\ref{aF}\right)$, we get
		\begin{align*}
			\frac{2\varepsilon}{\sqrt{m}}&\geq\frac{1}{m}\Vert \mathcal{A}[D(D^*\bar{H}D)D^*]\Vert_{1}-\frac{1}{m}\Vert \mathcal{A}[D(D^*HD -D^*\bar{H}D)D^*]\Vert_{1}.\\
			&\geq \alpha\Vert D^*\bar{H}D\Vert_{F}-\beta\left(\frac{4}{\sqrt{r}}+\frac{1}{r}\right)\Vert\left(D^*\bar{H}D\right)\Vert_{F}\\
			&=\left(\alpha-\frac{4\beta}{\sqrt{r}}-\frac{\beta}{r}\right)\Vert\left(D^*\bar{H}D\right)\Vert_{F}.
		\end{align*}
		Since
		\[
		\alpha-\frac{4\beta}{\sqrt{r}}-\frac{\beta}{r}>0,
		\]
		we have
		\begin{equation}\label{add:10}
			\Vert D^*\bar{H}D\Vert_{\mathit{F}}\leq\frac{2\varepsilon}{\sqrt{m}(\alpha-\frac{4\beta}{\sqrt{r}}-\frac{\beta}{r})}.
		\end{equation}
		
		Finally, substituting (\ref{add:d5}) and (\ref{add:10}) into (\ref{add:11}), we have
		\begin{align}\label{DHF}
			\Vert H\Vert_{\mathit{F}}\leq\left(\frac{1}{r}+\frac{4}{\sqrt{r}}+1\right) \Vert D^*\bar{H}D\Vert_{\mathit{F}}
			\leq\frac{2\varepsilon}{\sqrt{m}}\cdot\frac{\frac{1}{r}+\frac{4}{\sqrt{r}}+1}{\alpha-\frac{4\beta}{\sqrt{r}}-\frac{\beta}{r}}
			=\frac{2C\varepsilon}{\sqrt{m}},
		\end{align}	
		where
		\[
		C=\frac{\frac{1}{r}+\frac{4}{\sqrt{r}}+1}{\alpha-\frac{4\beta}{\sqrt{r}}-\frac{\beta}{r}}.
		\]
		Moreover, based on Lemma~\ref{xy}, we know
		\begin{align*}
			\mathop{\min}\limits_{c\in\mathbb{C},\vert c\vert=1}\Vert c\cdot{\hat{\textbf{x}}}-\textbf{x}_{0}\Vert_2
			&\leq\Vert {\hat{\textbf{x}}}-\textbf{x}_{0}\Vert_2\leq\frac{\sqrt{2}\Vert \hat{\mathbf{x}}\hat{\mathbf{x}}^*-\mathbf{x}_0\mathbf{x}_0^*\Vert_{\mathit{F}}}{\Vert \textbf{x}_0\Vert_{2}}\\
			&=\frac{\sqrt{2}\Vert H\Vert_{\mathit{F}}}{\Vert \textbf{x}_0\Vert_{2}}
			\leq\frac{2\sqrt{2}C\varepsilon}{\sqrt{m}\Vert \textbf{x}_{0}\Vert_2},
		\end{align*}
		which completes the proof.
	\end{proof}
	
	\begin{Remark}
		When $D=I$, our Theorem~\ref{recovery} reduces to Theorem~1.3 in \cite{XX2021}.
	\end{Remark}
	
	\section{$\ell_q$-DRIP}\label{sec3}
	
	It is well known that $\ell_q\;(0<q\le 1)$-minimization performs better in many aspects of phase retrieval than $\ell_1$-minimization (see \cite{CH2022,HX24,XZ23,LWL25}). This is because the $\ell_q$ quasi-norm is closer to the $\ell_0$-norm, the non-convex $\ell_q$ minimization yields a sparser solution than $\ell_1$-minimization; even with low sampling rate or low sparsity condition, it can still ensure high reconstruction accuracy. Moreover, it is possible to balance sparsity and computational complexity by adjusting the value of $q$. Therefore, in this section, we consider recovering the complex dictionary-sparse signals via $\ell_q\;(0<q\le 1)$-analysis minimization in noiseless setting:
	\begin{equation}\label{QDx}
		\mathop{\min} \limits_{\textbf{x}\in\mathbb{C}^n}\Vert D^*\textbf{x}\Vert_q^q \qquad {\mbox{s.t.}} \qquad \mathcal{A}(\mathbf{x}\mathbf{x}^*)=\mathcal{A}(\mathbf{x_0}\mathbf{x_0}^*),
	\end{equation}
	where $\mathcal{A}$ is defined the same as in (\ref{def:A}).
	
	Similar to $\ell_1$-analysis minimization, the $\ell_q\;(0<q\le1)$-analysis model can be expressed by using lifting technology:
	\begin{equation}\label{QDX}
		\mathop{\min} \limits_{X\in\mathbb{H}^{n\times n}}\Vert D^*XD\Vert_q^q \qquad {\mbox{s.t.}} \qquad \mathcal{A}(X)=\mathcal{A}(X_0),\;\;{\mbox{rank}}(X)=1,
	\end{equation}
	where $X=\mathbf{xx}^*$, and $X_0=\mathbf{x_0}\mathbf{x_0}^*$.

	Next, we extend the definition of $\ell_1$-DRIP proposed in (\ref{CDRIP}) for low-rank and dictionary-sparse matrices to $\ell_q$-analysis minimization (\ref{QDx}).
	
	\begin{Definition}[$\ell_q$-DRIP]\label{QDRIP}
		For a dictionary $\mathit{D}\in\mathbb{F}^{n\times N}$ and a map $\mathcal{A}$: $\mathbb{H}^{n\times n}\rightarrow \mathbb{R}^m$, $\mathcal{A}$ is called to have the $\ell_q$-DRIP of order $(s,k)$ with positive constants $\varphi$ and $\psi$, if
		\begin{equation}\label{QCDRIP}
			\varphi\lVert \mathit{DZD^*}\rVert_{\mathit{F}}^{q}\leq \lVert\mathcal{A}(\mathit{DZD^*})\rVert_{q}^q\leq \psi\lVert\mathit{DZD^*}\rVert_{\mathit{F}}^{q},
		\end{equation}
		holds for all $\mathit{Z}\in\mathbb{H}^{N\times N}$ with rank($Z$)$\leq s$ and $\lVert \mathit{Z}\rVert_{0,2}\leq k$.
	\end{Definition}
	
	\begin{Remark}
		If $m\gtrsim k+qk\log (n/k)$, then a Gaussian random map $\mathcal{A}$: $\mathbb{H}^{n\times n}\rightarrow \mathbb{R}^m$ satisfies the $\ell_q$-DRIP of order $(2,k)$ with high probability.
	\end{Remark}
	
	Finally, we derive a sufficient condition on quadratic measurements that can guarantee the unique recovery of dictionary-sparse signals via the $\ell_q$-analysis model.
	
	\begin{Theorem}\label{qrecovery}
		Let $D\in \mathbb{C}^{n\times N}$ be a dictionary, $\mathbf{x}_{0}\in\mathbb{C}^n$ be a dictionary $k$-sparse signal, and $0<q\le 1$. If the map $\mathcal{A}$ has the $\ell_q$-DRIP of order $(2, 2rk)$ with
		\begin{equation*}
			\varphi>\psi(\frac{1}{r^{2-q}}+\frac{2^{2+\frac{q}{2}}}{r^{1-\frac{q}{2}}})
		\end{equation*}
		holds for some sufficiently large $r>1$, then the solution $\tilde{\mathbf{x}}$ to the $\ell_q$-analysis minimization $(\ref{QDx})$ satisfies
		\begin{equation*}
			\tilde{\mathbf{x}}(\tilde{\mathbf{x}})^*=\mathbf{x}_0\mathbf{x}_0^*.
		\end{equation*}
	\end{Theorem}
	
	\begin{proof}
		Assume that the dictionary $D\in \mathbb{C}^{n\times N}$ is normalized, i.e., $DD^*=I_{n}$, and $\|D^*YD\|_F=\|Y\|_F$. Suppose that $\tilde{\mathbf{x}}$ is a solution to model (\ref{QDx}), then $e^{it}\tilde{\textbf{x}}$ is also a solution to model $(\ref{QDx})$ for all $t\in\mathbb{R}$. Hence, without loss of generality, set
		\begin{equation*}
			\left\langle D^*\tilde{\textbf{x}},D^*\textbf{x}_0\right\rangle\in \mathbb{R} \qquad {\mbox{and}} \qquad \left\langle D^*\tilde{\textbf{x}},D^*\textbf{x}_0\right\rangle\geq0.
		\end{equation*}
		Notice that the $\ell_q$-analysis model $(\ref{QDx})$ can be rewritten as model $\left(\ref{QDX}\right)$ by using lifting technique, and  $\tilde{X}$ is the solution to model $\left(\ref{QDX}\right)$ if and only if $\tilde{X}=\tilde{\textbf{x}}(\tilde{\textbf{x}})^\ast$.
		
		Let $H:=\tilde{\mathbf{x}}(\tilde{\mathbf{x}})^*-\textbf{x}_0{\textbf{x}_0}^\ast=\tilde{X}-X_0$. Denote $T_0$ as the support of ${D}^*\textbf{x}_0$, $T_1$ as the index set of the $rk$ largest entries of $({D}^*\tilde{\textbf{x}})_{T_0^c}$ in magnitude, $T_2$ as the index set of the second $rk$ largest entries of $({D}^*\tilde{\textbf{x}})_{T_0^c}$ in magnitude, and so on.  Let $T_{01}:=T_0\cup T_1$, $D^*\bar{H}D:=(D^*\tilde{\textbf{x}})_{T_{01}}(D^*\tilde{\textbf{x}})^*_{T_{01}}-(D^*\textbf{x}_0)_{T_{01}}(D^*\textbf{x}_0)^*_{T_{01}}$, and $\left( D^*HD\right)_{T_i,T_j}:=\left( D^*\tilde{\textbf{x}}\right)_{T_i}\left(D^*\tilde{\textbf{x}}\right)^*_{T_j}-\left( D^*{\textbf{x}_0}\right)_{T_i}\left(D^*{\textbf{x}_0}\right)^*_{T_j}$, for $i, j\ge 2$.
		
		By the construction of $T_j$, we have
		\begin{equation}\label{add:d1}
			\Vert\left(D^*\tilde{\mathbf{x}}\right)_{T_j}\Vert_2^q\leq\frac{\Vert\left(D^*\tilde{\mathbf{x}}\right)_{T_{j-1}}\Vert_q^q}{(rk)^{1-\frac{q}{2}}}, \;\; {\mbox{for}}\,\;j\geq 2.
		\end{equation}
		
		Note that
		\begin{align}
			\Vert H\Vert_F^q=&\Vert D^*HD\Vert_F^q\nonumber\\
			\leq&\Vert D^*\bar{H}D\Vert_F^q+\Vert\left(D^*HD\right)_{T_{01}^c,T_{01}^c}\Vert_F^q+2\sum_{i=0,1}\sum_{j\geq2}\Vert\left(D^*HD\right)_{T_{i},T_{j}}\Vert_F^q.\label{hf}
		\end{align}
		Next, we estimate the upper bounds for $\Vert D^*\bar{H}D\Vert_F^q$, $\Vert\left(D^*HD\right)_{T_{01}^c,T_{01}^c}\Vert_F^q$, and\\ $\sum_{i=0,1}\sum_{j\geq2}\Vert\left(D^*HD\right)_{T_{i},T_{j}}\Vert_F^q$, respectively.
		
		\textbf{Step 1}: Estimate the upper bound of $\Vert\left(D^*HD\right)_{T_{01}^c,T_{01}^c}\Vert_F^q$.
		
		Since $\tilde{\mathbf{x}}$ is a solution to the $\ell_q$-analysis minimization $(\ref{QDx})$, we have
		\begin{equation}\label{min1}
			\Vert D^*\tilde{X}D\Vert_q\leq\Vert D^*X_0D\Vert_q.
		\end{equation}
		Thus,
		\begin{align}\label{s2}
			\Vert\left(D^*HD\right)_{T_{0}^c,T_{0}^c}\Vert_q^q\leq&\Vert D^*HD-\left(D^*HD\right)_{T_{0},T_{0}}\Vert_q^q=\Vert D^*\tilde{X}D-(D^*\tilde{X}D)_{T_{0},T_{0}}\Vert_q^q\nonumber\\
			=&\Vert D^*\tilde{X}D\Vert_q^q-\Vert (D^*\tilde{X}D)_{T_{0},T_{0}}\Vert_q^q\overset{(\ref{min1})}{\leq}\Vert D^*X_0D\Vert_q^q-\Vert (D^*\tilde{X}D)_{T_{0},T_{0}}\Vert_q^q\nonumber\\
			\leq&\Vert D^*X_0D-(D^*\tilde{X}D)_{T_{0},T_{0}}\Vert_q^q=\Vert(D^*HD)_{T_{0},T_{0}}\Vert_q^q.
		\end{align}
		Therefore, we obtain
		\begin{align}\label{q:1}
			\Vert\left(D^*HD\right)_{T_{01}^c,T_{01}^c}\Vert_F^q\le&\sum_{i\geq2}\sum_{j\geq2}\Vert(D^*\tilde{X}D)_{T_{i},T_{j}}\Vert_F^q
			\leq\sum_{i\geq2}\sum_{j\geq2}\Vert\left(D^*\tilde{\mathbf{x}}\right)_{T_i}\Vert_2^q\Vert\left(D^*\tilde{\mathbf{x}}\right)_{T_j}\Vert_2^q\nonumber\\
			\overset{(\ref{add:d1})}{\leq}&\frac{1}{(rk)^{2-q}}\left(\sum_{i\geq1}\Vert(D^*\tilde{\mathbf{x}})_{T_i}\Vert_q^q\right)^2
			=\frac{1}{(rk)^{2-q}}\Vert\left(D^*\tilde{\mathbf{x}}\right)_{T_0^c}\Vert_q^{2q}\nonumber\\
			=&\frac{1}{(rk)^{2-q}}\Vert\left(D^*HD\right)_{T_{0}^c,T_{0}^c}\Vert_q^q
			\overset{(\ref{s2})}{\leq}\frac{1}{(rk)^{2-q}}\Vert(D^*HD)_{T_{0},T_{0}}\Vert_q^q\nonumber\\
			\leq&\frac{1}{r^{2-q}}\Vert(D^*HD)_{T_{0},T_{0}}\Vert_F^q\le \frac{1}{r^{2-q}}\Vert D^*\bar{H}D\Vert_{F}^q.
		\end{align}
		
		\textbf{Step 2}: Estimate the upper bound of $\sum_{i=0,1}\sum_{j\geq2}\Vert\left(D^*HD\right)_{T_{i},T_{j}}\Vert_F^q$.
		
		For $i\in\{0,1\}$, we have
		\begin{align}\label{q:2}
			&\sum_{j\geq2}\Vert\left(D^*HD\right)_{T_{i},T_{j}}\Vert_F^q=\sum_{j\geq2}\Vert\left(D^*HD\right)_{T_{j},T_{i}}\Vert_F^q
			=\sum_{j\geq2}\Vert\left(D^*\tilde{\mathbf{x}}\right)_{T_{i}}\left(D^*\tilde{\mathbf{x}}\right)_{T_{j}}^*\Vert_F^q\nonumber\\
			\leq&\Vert\left(D^*\tilde{\mathbf{x}}\right)_{T_{i}}\Vert_2^q\sum_{j\geq2}\Vert\left(D^*\tilde{\mathbf{x}}\right)_{T_{j}}\Vert_2^q
			\overset{(\ref{add:d1})}{\leq}\Vert\left(D^*\tilde{\mathbf{x}}\right)_{T_{i}}\Vert_2^q\cdot\frac{1}{(rk)^{1-\frac{q}{2}}}\Vert\left(D^*\tilde{\mathbf{x}}\right)_{T_{0}^c}\Vert_q^q\nonumber\\
			\leq&\frac{1}{(rk)^{1-\frac{q}{2}}}\Vert\left(D^*\tilde{\mathbf{x}}\right)_{T_{01}}\Vert_2^q\cdot\Vert\left(D^*\tilde{\mathbf{x}}\right)_{T_{0}^c}\Vert_q^q.
		\end{align}
		Since $\tilde{\mathbf{x}}$ is a solution to the $\ell_q$-analysis minimization $(\ref{QDx})$, yields
		\begin{align}
			\Vert\left(D^*\tilde{\mathbf{x}}\right)_{T_{0}^c}\Vert_q^q
			\leq&\Vert\left(D^*\mathbf{x}_0\right)\Vert_q^q-\Vert\left(D^*\tilde{\mathbf{x}}\right)_{T_{0}}\Vert_q^q
			\leq \Vert D^*\mathbf{x}_0-\left(D^*\tilde{\mathbf{x}}\right)_{T_{0}}\Vert_q^q\nonumber\\
			\leq& k^{1-\frac{q}{2}}\Vert D^*\mathbf{x}_0-\left(D^*\tilde{\mathbf{x}}\right)_{T_{0}}\Vert_2^q
			\leq k^{1-\frac{q}{2}}\Vert D^*\mathbf{x}_0-\left(D^*\tilde{\mathbf{x}}\right)_{T_{01}}\Vert_2^q,\label{mq1}
		\end{align}
		Substituting (\ref{mq1}) into (\ref{q:2}), we obtain
		\begin{align}\label{mq2}
			&\sum_{j\geq2}\Vert\left(D^*HD\right)_{T_{i},T_{j}}\Vert_F^q
			\leq\frac{1}{(rk)^{1-\frac{q}{2}}}\Vert\left(D^*\tilde{\mathbf{x}}\right)_{T_{01}}\Vert_2^q\cdot\Vert\left(D^*\tilde{\mathbf{x}}\right)_{T_{0}^c}\Vert_q^q\nonumber\\
			\leq&\frac{k^{1-\frac{q}{2}}}{(rk)^{1-\frac{q}{2}}}\Vert\left(D^*\tilde{\mathbf{x}}\right)_{T_{01}}\Vert_2^q\cdot \Vert D^*\mathbf{x}_0-\left(D^*\tilde{\mathbf{x}}\right)_{T_{01}}\Vert_2^q\nonumber\\
			\leq&\frac{2^{\frac{q}{2}}}{r^{1-\frac{q}{2}}}\Vert D^*X_0D-(D^*\tilde{X}D)_{T_{01},T_{01}}\Vert_F^q=\frac{2^{\frac{q}{2}}}{r^{1-\frac{q}{2}}}\Vert D^*\bar{H}D\Vert_F^q,\;\;{\mbox{for}}\;\; i=0,1,
		\end{align}
		where we use Lemma~\ref{xy} in the last inequality.
		
		\textbf{Step 3}: Estimate the upper bound of $\Vert D^*\bar{H}D\Vert_F^q$.
		
		Since
		\begin{eqnarray*}
			\Vert\mathcal{A}[D(D^*HD)D^*]\Vert_{2}=\Vert\mathcal{A}(H)\Vert_{2}=\Vert\mathcal{A}(\tilde{X})-\mathcal{A}(X_0)\Vert_{2}=0,
		\end{eqnarray*}
		we know
		\begin{align}\label{}
			0&=\Vert \mathcal{A}[D(D^*HD)D^*]\Vert_{q}^q\nonumber\\
			&\geq\Vert \mathcal{A}[D(D^*\bar{H}D)D^*]\Vert_{q}^q-\Vert \mathcal{A}[D(D^*HD -D^*\bar{H}D)D^*]\Vert_{q}^q.
		\end{align}
		Hence,
		\begin{align}
			&\Vert\mathcal{A}[D(D^*\bar{H}D)D^*]\Vert_q^q\le\Vert\mathcal{A}[D(D^*HD-D^*\bar{H}D)D^*]\Vert_q^q\nonumber\\
			\le&\Vert\mathcal{A}[D((D^*HD)_{T_0,T_{01}^c}+(D^*HD)_{T_{01}^c,T_0})D^*]\Vert_q^q\nonumber\\
			\quad&+\Vert\mathcal{A}[D((D^*HD)_{T_1,T_{01}^c}
			+(D^*HD)_{T_{01}^c,T_1})D^*]\Vert_q^q+\Vert\mathcal{A}[D(D^*HD)_{T_{01}^c,T_{01}^c}D^*]\Vert_q^q\nonumber\\
			\le&2\sum_{j\ge 2}\Vert\mathcal{A}[D(D^*HD)_{T_0,T_{j}}D^*]\Vert_q^q
			+2\sum_{j\ge 2}\Vert\mathcal{A}[D(D^*HD)_{T_1,T_{j}}D^*]\Vert_q^q\nonumber\\
			\quad&+\sum_{i\ge 2}\sum_{j\ge 2}\Vert\mathcal{A}[D(D^*HD)_{T_{i},T_{j}}D^*]\Vert_q^q.\label{dq}
		\end{align}
		Applying the $\ell_q$-DRIP property (\ref{QCDRIP}) to (\ref{dq}), we have
		\begin{align*}
			\varphi\Vert D^*\bar{H}D\Vert_q^q=&\varphi\Vert D(D^*\bar{H}D)D^*\Vert_q^q\le\Vert\mathcal{A}[D(D^*\bar{H}D)D^*]\Vert_q^q\\
			\le&2\psi\sum_{j\ge 2}\Vert D(D^*HD)_{T_0,T_{j}}D^*\Vert_F^q
			+2\psi\sum_{j\ge 2}\Vert D(D^*HD)_{T_1,T_{j}}D^*\Vert_F^q\nonumber\\
			\quad&+\psi\sum_{i\ge 2}\sum_{j\ge 2}\Vert D(D^*HD)_{T_{i},T_{j}}D^*\Vert_F^q\\
			\leq&\psi\left(2\sum_{i=0,1}\sum_{j\geq2}\Vert(D^*HD)_{T_i,T_j}\Vert_{F}^q+\sum_{i\geq2}\sum_{j\geq2}\Vert(D^*HD)_{T_i,T_j}\Vert_{F}^q\right)\\
			\leq&\psi\left(\frac{1}{r^{2-q}}+\frac{2^{2+\frac{q}{2}}}{r^{1-\frac{q}{2}}}\right)\Vert D^*\bar{H}D\Vert_{F}^q,
		\end{align*}
		where we use (\ref{q:1}) and (\ref{mq2}) in the last step.
		Since
		\begin{equation*}
			\varphi>\psi(\frac{1}{r^{2-q}}+\frac{2^{2+\frac{q}{2}}}{r^{1-\frac{q}{2}}}),
		\end{equation*}
		we get 
		\begin{equation}\label{dhd}
			\Vert D^*\bar{H}D\Vert_F^q=0.
		\end{equation}
		
		Finally, substituting (\ref{q:1}), (\ref{mq2}) and (\ref{dhd}) into (\ref{hf}), we have
		\begin{align*}
			\Vert H\Vert_F^q
			\leq&\Vert\left(D^*HD\right)_{T_{01}^c,T_{01}^c}\Vert_F^q+2\sum_{i=0,1}\sum_{j\geq2}\Vert\left(D^*HD\right)_{T_{i},T_{j}}\Vert_F^q+\Vert D^*\bar{H}D\Vert_F^q\\
			\le&\frac{1}{r^{2-q}}\Vert D^*\bar{H}D\Vert_{F}^q+\frac{2^{\frac{q}{2}}}{r^{1-\frac{q}{2}}}\Vert D^*\bar{H}D\Vert_F^q+\Vert D^*\bar{H}D\Vert_F^q\\
			=&\left(\frac{1}{r^{2-q}}+\frac{2^{\frac{q}{2}}}{r^{1-\frac{q}{2}}}+1\right)\Vert D^*\bar{H}D\Vert_F^q=0,
		\end{align*}
		which completes the proof.
	\end{proof}
	
	\begin{Remark}
		When $D=I$, our Theorem~\ref{qrecovery} reduces to Theorem~3.1 in \cite{XZ23}.
	\end{Remark}
	
	\section{Conclusion}\label{sec:C}
	
	In \cite{G2016}, Gao has shown that S-DRIP can guarantee a stable recovery of real dictionary-sparse signals from phaseless measurements. However, it remains unknown whether this conclusion also holds for complex signals. As a supplement to \cite{G2016}, our paper fills in this research gap. In this paper, we study the recovery of complex dictionary-sparse signals from the magnitude-only measurements via both $\ell_1$-analysis minimization and its generalized $\ell_q\;(0<q\leq1)$-analysis model. First, we introduce a sufficient condition $\ell_1$-DRIP on quadratic measurements, which can ensure stable reconstruction of complex dictionary-sparse signals by solving $\ell_1$-analysis minimization. Then, we extend the $\ell_1$-DRIP to $\ell_q\;(0<q\leq1)$-analysis model without the noise. However, there are still many unresolved issues in the field of dictionary-sparse phase retrieval, for example, the lack of effective reconstruction algorithms. One of our future directions is to propose some dictionary-sparse phase retrieval reconstruction algorithms with low computational complexity.
	
	%

	\section*{Acknowledgements}
	This work is supported partially by the National Natural Science Foundation of China (Grant no. 12261059), and the Natural Science Foundation of Jiangxi Province (Grant no. 20224BAB211001).

	\section*{Conflict of Interest}
	This work does not have any conflicts of interest.
	

\end{document}